\documentclass[a4paper,twocolumn,11pt,accepted=2021-07-02]{quantumarticle}
\pdfoutput=1
\usepackage[utf8]{inputenc}
\usepackage[english]{babel}
\usepackage[table]{xcolor}
\usepackage[T1]{fontenc}
\usepackage{amsmath}
\usepackage{hyperref}
\usepackage{tikz}
\usepackage{lipsum}
\usepackage{graphicx}
\usepackage{epsfig}
\usepackage{epstopdf}
\usepackage{braket}
\usepackage{hyperref}
\usepackage{amsthm}
\usepackage{enumitem}
\usepackage{dsfont}
\usepackage{bbold}
\usepackage{bm}
\usepackage{eurosym}
\usepackage{diagbox}
\usepackage[numbers,sort&compress]{natbib}
\usepackage[ruled,vlined]{algorithm2e}

\newcommand{\be}{\begin{equation}}
\newcommand{\ee}{\end{equation}}
\newcommand{\ba}{\begin{aligned}}
\newcommand{\ea}{\end{aligned}}

\newcommand{\bc}{\begin{center}}
\newcommand{\ec}{\end{center}}
\newcommand{\beq}{\begin{equation}}
\newcommand{\eeq}{\end{equation}}
\newcommand{\beqq}{\begin{equation*}}
\newcommand{\eeqq}{\end{equation*}}
\newcommand{\beqa}{\begin{align}}
\newcommand{\eeqa}{\end{align}}
\newcommand{\barr}{\begin{array}}
\newcommand{\earr}{\end{array}}
\newcommand{\bi}{\begin{itemize}}
\newcommand{\ei}{\end{itemize}}

\newtheorem{lem}{Lemma}

\newtheorem{defi}{Definition}
\DeclareMathOperator{\poly}{poly\,}

\DeclareMathOperator{\Tr}{Tr}

\DeclareMathOperator{\Per}{Per}

\newcommand*{\vv}[1]{\vec{#1}}

\begin{document}

\title{Quantum machine learning with adaptive linear optics}

\author{Ulysse Chabaud}
\email{ulysse.chabaud@gmail.com}
\orcid{0000-0003-0135-9819}
\affiliation{Department of Computing and Mathematical Sciences, California Institute of Technology}
\affiliation{Universit\'e de Paris, IRIF, CNRS, France}
\affiliation{Laboratoire d'Informatique de Paris 6, CNRS, Sorbonne Universit\'e, 4 place Jussieu, 75005 Paris, France}
\author{Damian Markham}
\affiliation{Laboratoire d'Informatique de Paris 6, CNRS, Sorbonne Universit\'e, 4 place Jussieu, 75005 Paris, France}
\affiliation{JFLI, CNRS, National Institute of Informatics, University of Tokyo, Tokyo, Japan}
\author{Adel Sohbi}
\email{sohbi@kias.re.kr}
\affiliation{School of Computational Sciences, Korea Institute for Advanced Study, Seoul 02455, Korea}
\orcid{0000-0002-1275-9722}
\maketitle


\begin{abstract}
\noindent We study supervised learning algorithms in which a quantum device is used to perform a computational subroutine---either for prediction via probability estimation, or to compute a kernel via estimation of quantum states overlap. We design implementations of these quantum subroutines using Boson Sampling architectures in linear optics, supplemented by adaptive measurements. We then challenge these quantum algorithms by deriving classical simulation algorithms for the tasks of output probability estimation and overlap estimation. We obtain different classical simulability regimes for these two computational tasks in terms of the number of adaptive measurements and input photons. In both cases, our results set explicit limits to the range of parameters for which a quantum advantage can be envisaged with adaptive linear optics compared to classical machine learning algorithms: we show that the number of input photons and the number of adaptive measurements cannot be simultaneously small compared to the number of modes. Interestingly, our analysis leaves open the possibility of a near-term quantum advantage with a single adaptive measurement.
\end{abstract}


\section{Introduction}

Quantum computers promise dramatic advantages over their classical counterparts~\cite{feynman1982simulating,shor1994algorithms}, but a fault-tolerant universal quantum computer is still far from being available~\cite{preskill2018quantum}. The quest for near-term quantum speedup has thus led to the introduction of various subuniversal models---models that are believed to have an intermediate computational power between classical and universal quantum computing---such as Boson Sampling~\cite{Aaronson2013} or IQP circuits~\cite{Bremmer2010}, recently culminating with the experimental demonstration of random circuit sampling~\cite{arute2019quantum} and Gaussian Boson Sampling~\cite{zhong2020quantum}.

Finding practical applications for these subuniversal models, other than the demonstration of quantum speedup, is a timely issue, as it may enable interesting quantum advantages in the era of Noisy Intermediate-Scale Quantum devices~\cite{preskill2018quantum}. 

Recently, there has been an increased interest on the possibility of enhancing classical machine learning algorithms using quantum computers~\cite{harrow2009quantum}, which includes the development of quantum neural networks \cite{Schuld2014, Romero_2017, doi:10.1098/rspa.2017.0551, Dunjko_2018, Zoufal2019, PhysRevResearch.1.033063, farhi2018classification, Beer2020, abbas2020power, cerezo2020variational} and the development of quantum kernel methods \cite{havlivcek2019supervised, schuld2019quantum, Blank2020, Bartkiewicz2020, liu2020rigorous, huang2020power}. 

In particular, recent proposals have been driven by subuniversal models such as Gaussian Boson Sampling~\cite{hamilton2017gaussian,schuld2019quantum2} or IQP circuits~\cite{Bremmer2010,havlivcek2019supervised}. In the latter, the authors considered supervised learning algorithms in which some computational subroutines are executed in a quantum way, namely the estimation of the output probabilities of quantum circuits, or the estimation of the overlap of the output states of quantum circuits. They showed that IQP circuits alone could not provide a quantum advantage for these subroutines and therefore considered minimal extensions of these circuits, in terms of circuit depth.

Hereafter, we study the use of Boson Sampling linear optical interferometers~\cite{Aaronson2013}, with input photons, for similar quantum machine learning tasks. Instead of extending the depth which was shown to improve machine learning model performances \cite{havlivcek2019supervised} while making the training phase challenging \cite{McClean_2018,cerezo2020costfunctiondependent, wang2020noiseinduced}, we allow for adaptive measurements---intermediate measurements that drive the rest of the computation---which provide a natural analogy with the circuit depth in the linear optics picture~\cite{briegel2009measurement}: by encoding qubits into single-photons and using a sufficient number of adaptive measurements, one can perform universal quantum computing~\cite{knill2001scheme}. 

We give a detailed prescription for performing quantum machine learning with classical data, using adaptive linear optics for computational subroutines such as probability estimation and overlap estimation.

We also examine the classical simulability of these quantum subroutines. More precisely, we give classical simulation algorithms whose runtimes are explicitly dependent on: (i) the number of modes $m$, (ii) the number of adaptive measurements $k$, (iii) the number of input photons $n$ and (iv) the number of photons $r$ detected during the adaptive measurements. This effectively sets a limit on the range of parameters for which adaptive linear optics may provide an advantage for machine learning over classical computers using our methods, thus identifying the regimes where a quantum advantage can be envisaged. Boson Sampling instances~\cite{Aaronson2013} correspond to the case with no adaptive measurement, while the Knill--Laflamme--Milburn scheme for universal quantum computing~\cite{knill2001scheme} corresponds to the case where the number of adaptive measurements scales linearly with the size of the computation. 

For probability estimation, we show that the classical simulation is efficient whenever the number of adaptive measurements or the number of input photons is constant. Moreover, the number of input photons and the number of adaptive measurements cannot be simultaneously small compared to the number of modes (see Table~\ref{tab:probaest}).

For overlap estimation, we show a similar behaviour, although in this case our results do not rule out the possibility of a quantum advantage with a single adaptive measurement (see Tables~\ref{tab:overlapest} and~\ref{tab:overlapest2}). Our main technical contribution is an expression for the inner product of the output states of two adaptive unitary interferometers which is essentially independent of the number of adaptive measurements.

The rest of the paper is organised as follows.
In section~\ref{sec:background}, we provide a background on quantum machine learning with classical data and classical simulation of quantum computations. In section~\ref{sec:ALO}, we introduce the model of adaptive linear optics which we consider. We give a prescription for performing probability estimation and overlap estimation with instances of this model, and detail how to use these as subroutines for machine learning problems.
In section~\ref{sec:simulation}, we derive two classical simulation algorithms, one for each of these two tasks, and we analyse the running time of these algorithms. We conclude in section~\ref{sec:conclusion}.


\section{Background}
\label{sec:background}


\subsection{Kernel methods for quantum machine learning}


\subsubsection{Encoding classical data with quantum states}

We consider a typical machine learning problem, such as a classification problem, where a classical dataset $\mathcal{D}=\{\vv x_1,\dots,\vv x_{\vert \mathcal{D} \vert}\}$ from an input set $\mathcal{X}$ is given.
One way to use quantum computers to solve such problems is to encode classical data onto quantum states such that there exists a so-called feature map $\vv x_l \mapsto \vert \phi(\vv x_l)\rangle$ which can be processed by a quantum computer.

\begin{defi}[Feature map \cite{schuld2019quantum}]\label{def:fMap}
Let $\mathcal{F}$ be a Hilbert space, called feature Hilbert space, $\mathcal{X}$ a non empty set, called input set, and $\vv x\in\mathcal{X}$. A feature map is a map $\phi:\mathcal{X}\rightarrow\mathcal{F}$ from inputs to vectors in the Hilbert space.
\end{defi}

\noindent Many machine learning algorithms perform well in linear cases such as the support vector machines which will be used hereafter. However, many real world problems require non-linearity to make successful predictions. By using \textit{kernel methods} one can introduce non-linearity and use estimation methods that are linear in terms of the kernel evaluations.

\begin{defi}[Kernel \cite{schuld2019quantum}]\label{def:kernel}
Let $\mathcal{X}$ be an input set. A function $\kappa: \mathcal{X} \times \mathcal{X} \rightarrow \mathbb{C}$ is called a kernel if for any finite subset $\mathcal{D}=\{\vv x_1,\dots,x_{\vert \mathcal{D} \vert}\}$ with $M\geq2$ the Gram matrix $K$ with entries $K_{l,l'}=\kappa(\vv x_l,\vv x_{l'})$ is positive semidefinite.
\end{defi}

\noindent The kernel corresponds to a dot product in a feature space (here in a high-dimensional Hilbert space). In \cite{schuld2021supervised}, it is shown that the notions of feature map in Hilbert space and kernel can be connected. A straightforward way is to define a (quantum) kernel $K$ from a feature map $\phi$ as follows:

\begin{equation}\label{eq:fMapKernel}
\kappa(\vv x_l,\vv x_{l'}) = |\langle \phi(\vv x_l) \vert \phi(\vv x_{l'}) \rangle_\mathcal{F}|^2,
\end{equation}
where $x_m \in \mathcal{X}$, $x_{m'} \in \mathcal{X}$ and $\langle \cdot| \cdot\rangle_\mathcal{F}$ is the inner product over the Hilbert space $\mathcal{F}$.


\subsubsection{Using Feature Hilbert Spaces for Machine Learning}

There are two main ways to use feature Hilbert spaces:

\begin{itemize}
\item The \emph{quantum variational classification} \cite{havlivcek2019supervised} or \emph{explicit approach} \cite{schuld2019quantum}:
the entire model computation is preformed on a quantum device trained by a hybrid variational quantum-classical or classical algorithm \cite{PhysRevA.98.032309, farhi2018classification, PhysRevResearch.1.033063}. In this case, the probability distribution over the possible outcomes is used for the classification. Hence, while the feature map is used to encode the data, the kernel is not known directly from the quantum device. In \cite{havlivcek2019supervised}, the probability to obtain a certain binary output $b_i$ for the classical data input $x$ is given by:

\begin{equation}\label{eq:ibmprob}
	\Pr\,(b_i) = \vert \langle b_i \vert W(\theta) U_{\phi(x)} \vert 0^{\otimes n} \rangle \vert ^2,
\end{equation}
where $U_{\phi(x)}$ encodes the feature map and $W(\theta)$ corresponds to the quantum support vector machine.

\item The \emph{quantum kernel estimation} \cite{havlivcek2019supervised} or \emph{implicit approach} \cite{schuld2019quantum}: a quantum-assisted method where the quantum device is only used to evaluate the kernel and the rest of the machine learning algorithm is carried on by a classical algorithm. In \cite{havlivcek2019supervised}, the kernel between two classical data vectors $x_l$ and $x_{l'}$ is given by:
\begin{equation}\label{eq:ibmkernel}
	K_{l,l'} = |\langle 0^{\otimes n} \vert U_{\phi(x_l)}^\dagger U_{\phi(x_{l'})} \vert 0^{\otimes n} \rangle|^2.
\end{equation}
\end{itemize}

\noindent In order to challenge and justify the use of a quantum-assisted method, it is important to consider how hard it would be for a classical machine to compute the same quantities. Indeed, if the output probability or kernel could be estimated directly classically, there would be no need for a quantum computer for this task. Hence, when it comes to classical hardness, the first method requires that estimating the output probability in Eq.~\eqref{eq:ibmprob} is hard, while the second method requires that estimating the overlap in Eq.~\eqref{eq:ibmkernel} is hard (for classical computers). We give more formal definitions for these classical simulation tasks in the following section.


\subsection{Classical simulation of quantum computations}

Depending on the approach used for simulating classically the functioning of quantum devices, several notions of simulability are commonly used. One could ask the classical simulation algorithm to mimic the output of the quantum computation~\cite{terhal2002classical,pashayan2020estimation}: informally, a quantum computation is \textit{weakly simulable} if there exists a classical algorithm which outputs samples from its output probability distribution in time polynomial in the size of the quantum computation. Various relaxations of this definition are possible, allowing the classical sampling to be approximate rather than exact, or to abort with a small probability. The existence of such an efficient classical simulation has been ruled out for various subuniversal models of quantum computing, such as IQP circuits~\cite{Bremmer2010} or Boson Sampling~\cite{Aaronson2013}, under complexity-theoretic conjectures---both in the exact case and the approximate case, up to additional conjectures.

While weak simulation of quantum computations is arguably the most commonly studied, other notions of classical simulation may be useful: if the output samples of a quantum computation are used to compute a quantity which may be computed efficiently classically by other means, it is no longer necessary to simulate the whole quantum device. We consider two concrete examples which are prominent for variational quantum algorithms in quantum machine learning: probability estimation and overlap estimation~\cite{havlivcek2019supervised,schuld2019quantum}.

\begin{defi}[Probability estimation]
Let $P$ be a probability distribution over $M$ outcomes and let $\epsilon,\delta>0$. Given any outcome $\bm x$ in the sample space of $P$, \textit{probability estimation} refers to the computational task of outputting an estimate $\tilde P[\bm x]$ such that
\be
\left|\tilde P[\bm x]-P[\bm x]\right|\le\epsilon,
\ee
with probability greater than $1-\delta$, in time $O(\poly(\frac M\epsilon)\log\frac1\delta)$.
\end{defi}

\noindent Efficient probability estimation thus amounts to outputting an estimate of the probability of a fixed outcome with polynomially small additive error, with exponentially small probability of failure, in polynomial time in the size of the computation. One may use the samples from a quantum computation in order to perform efficiently probability estimation for any given outcome: given a quantum device of size $M$ which outputs samples from some probability distribution and a fixed outcome $\bm x$ in the sample space, one may run the device $\poly(M)$ times, recording the value $1$ whenever the outcome $\bm x$ is obtained and the value $0$ otherwise. Then, the frequency of the outcome $\bm x$ over the $\poly(M)$ uses of the quantum device is a polynomially precise additive estimate of the probability of the outcome $\bm x$ with exponentially small probability of failure, by virtue of Hoeffding inequality \cite{hoeffding1963probability}.

Weak simulation is at least as hard as probability estimation, since by the previous reasoning one may obtain polynomially precise additive estimates of probabilities from samples of the probability distribution. Moreover, there are some quantum computations for which weak simulation is hard for classical computers (assuming widely believed conjectures from complexity theory), but probability estimation can be done efficiently classically. This is the case for IQP circuits~\cite{Bremmer2010,havlivcek2019supervised}, Boson Sampling interferometers~\cite{Aaronson2013} and even the circuit implementing the period-finding subroutine of Shor's factoring algorithm~\cite{shor1994algorithms}. We detail the latter case in Appendix~\ref{app:Shor}, for the sake of clarifying the relations between these different notions of classical simulation. Note also that computing exactly output probabilities of quantum circuits or interferometers is a \#\textsf{P}-hard problem, therefore expected  to be hard classically and quantumly~\cite{van2010classical}.

A more general computational task than probability estimation in the context of quantum computing is the following:

\begin{defi}[Overlap estimation]
Let $\ket\phi$ and $\ket\psi$ be output states of two efficiently describable quantum computations of size $M$ and let $\epsilon,\delta>0$. \textit{Overlap estimation} refers to the computational task of outputting an estimate $\tilde O$ such that
\be
\left|\tilde O-|\braket{\phi|\psi}|^2\right|\le\epsilon,
\ee
with probability greater than $1-\delta$, in time $O(\poly(\frac M\epsilon)\log\frac1\delta)$.
\end{defi}

\noindent The overlap between two quantum states is a measure of their distinguishability~\cite{dieks1988overlap} and overlap estimation thus is related to quantum state discrimination. Several techniques exist to perform quantumly the overlap estimation of two states $\ket\phi$ and $\ket\psi$~\cite{fanizza2020beyond}. One of them is to perform the swap test~\cite{buhrman2001quantum} with various copies of both states.

Overlap estimation can be also done efficiently classically for IQP circuits. Nonetheless, this family of circuits has been identified as a promising venue for implementing quantum machine learning algorithms~\cite{havlivcek2019supervised}, when enlarged to contain similar circuits with bigger depth. Motivated by this approach, we consider hereafter the case of another subuniversal model: Boson Sampling~\cite{Aaronson2013} with input photons, supplemented with adaptive measurements.


\section{Quantum machine learning with adaptive linear optics}
\label{sec:ALO}

In what follows, we study Boson Sampling architectures~\cite{Aaronson2013}, supplemented with a given number of adaptive measurements---that is, some of the modes are measured throughout the computation and the rest of the computation can depend on their outcome---which we refer to as adaptive linear optics. In this section, we derive quantum algorithms for performing probability and overlap estimation with adaptive linear optics, together with a prescription for utilising these algorithms as subroutines in supervised learning algorithms.


\subsection{Adaptive linear optics}
\label{sec:ALO1}

Hereafter, we detail the computational model of adaptive linear optics which we consider. We first introduce a few notations.

In the linear optics picture, the input states we consider are multimode photon number Fock states over $m$ modes (we use bold math for multi-index notations, see Table~\ref{tab:bold}):
\be
\ket{\bm s}=\frac1{\sqrt{\bm s!}}\hat a_1^{\dag s_1}\dots\hat a_m^{\dag s_m}\ket0^{\otimes m},
\label{Fockm}
\ee
where $s_i$ and $\hat a_i^\dag$ are respectively the number of photons and the creation operator for the $i^{th}$ mode~\cite{Leonhardt-essential}. We identify these states with $m$-tuples of integers $\bm s=(s_1,\dots,s_m)\in\mathbb N^m$. Following~\cite{Aaronson2013}, let us define, for all $n\in\mathbb N$,
\be
\Phi_{m,n}:=\{\bm s=(s_1,\dots,s_m)\in\mathbb N^m\;|\;|\bm s|=n\}.
\ee
This set corresponds to the $m$-mode Fock states with total number of photons equal to $n$, and we have $|\Phi_{m,n}|=\binom{m+n-1}{n}$.

\begin{table}
\centering
\setlength\tabcolsep{0.1pt}
\bgroup
\def\arraystretch{2}
\begin{tabular}{| c | c |}
\hline
$|\bm s|$&$s_1+\cdots+s_m$\\
$\bm s!$&$s_1!\cdots s_m!$\\
$\ket{\bm s}$&$\ket{s_1\dots s_m}$\\
$\quad\bm s+\bm t\quad$&$\quad(s_1+t_1,\dots,s_m+t_m)\quad$\\
$\bm0^m$&$(0,\dots,0)$\\
$\bm1^m$&$(1,\dots,1)$\\
\hline
\end{tabular}
\egroup
\caption{Multi-index notations. For $m\in\mathbb N^*$, we write $\bm s=(s_1,\dots,s_m)\in\mathbb N^m$ and $\bm t=(t_1,\dots,t_m)\in\mathbb N^m$.}
\label{tab:bold}
\end{table}

We consider a unitary interferometer of size $m$, described by an $m\times m$ unitary matrix $U=(u_{ij})_{1\le i,j\le m}$. Unlike in the circuit picture, the matrix $U$ does not act on the computational basis, which is in this case the infinite multimode Fock basis, but rather describes the linear evolution of the creation operator of each mode. More precisely,
\be
\begin{pmatrix} \hat a_1^\dag \\ \vdots\\ \hat a_m^\dag \end{pmatrix}\mapsto U\begin{pmatrix} \hat a_1^\dag \\ \vdots\\ \hat a_m^\dag\end{pmatrix}=\begin{pmatrix} \sum_{k=1}^m{u_{1k}\hat a_k^\dag} \\ \vdots\\ \sum_{k=1}^m{u_{mk}\hat a_m^\dag}\end{pmatrix}.
\label{evocrea}
\ee
We write $\hat U$ instead the unitary action of the interferometer on the multimode Fock basis. Because the interferometer conserves the total number of photons, for all $p,q\in\mathbb N$, all $\bm s\in\Phi_{m,p}$ and all $\bm t\in\Phi_{m,q}$
\be
\braket{\bm s|\hat U|\bm t}=0
\ee
whenever $p\neq q$.
Let $n\in\mathbb{N}$, $\bm s=(s_1,\dots,s_m)\in\Phi_{m,n}$ and $\bm t=(t_1,\dots,t_m)\in\Phi_{m,n}$. Combining Eq.~\eqref{Fockm} and Eq.~\eqref{evocrea} we obtain~\cite{Aaronson2013}
\be
\braket{\bm s|\hat U|\bm t}=\frac{\Per(U_{\bm s,\bm t})}{\sqrt{\bm s!}\sqrt{\bm t!}},
\label{transf}
\ee
where $U_{\bm s,\bm t}$ is the $n\times n$ matrix obtained from $U$ by repeating $s_i$ times its $i^{th}$ row and $t_j$ times its $j^{th}$ column for $i,j=1,\dots,m$, and where the permanent of a $r\times r$ matrix $A=(a_{ij})_{1\le i,j\le r}$ is defined as
\be
\Per A=\sum_{\sigma\in\mathcal{S}_r}{\prod_{i=1}^r{a_{i\sigma(i)}}},
\label{per}
\ee
where $\mathcal{S}_r$ is the symmetric group over $\{1,\dots,r\}$.

\begin{figure*}[t]
	\begin{center}
		\includegraphics[width=1.9\columnwidth]{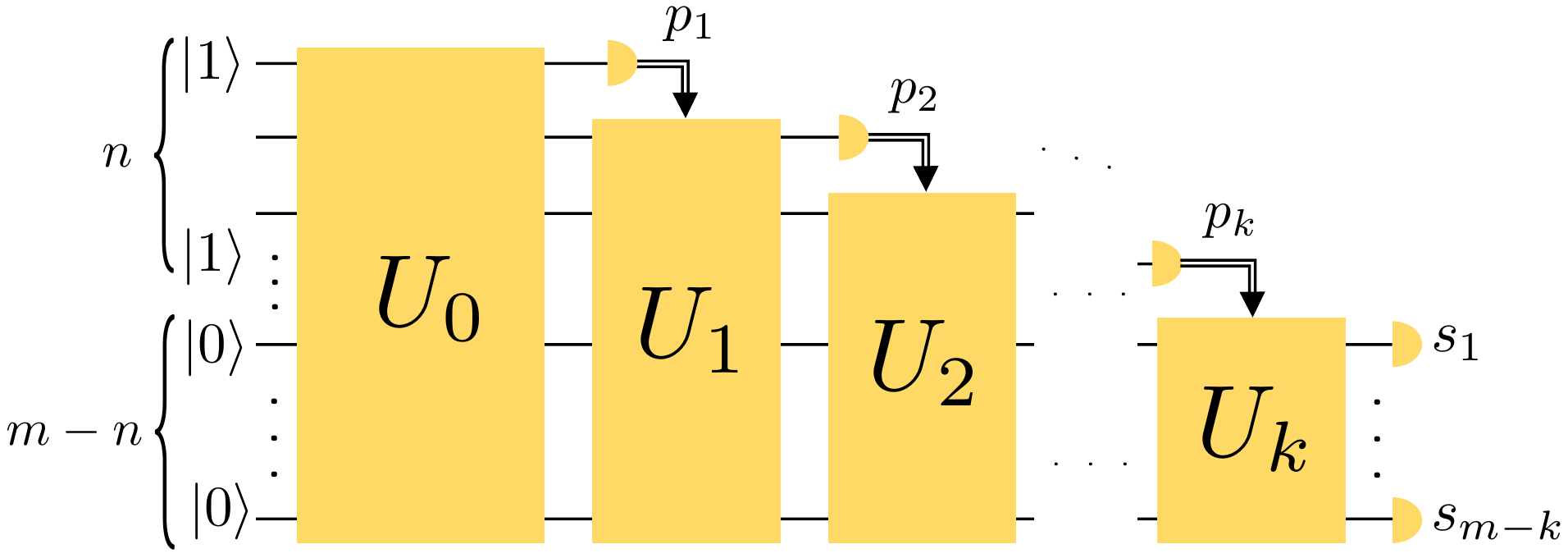}
		\caption{Linear optical computing model with $k$ adaptive measurements and input state $\ket{1\dots10\dots0}$, with $n$ photons over $m$ modes. The output modes are measured using photon number detection. For all $j\in\{1,\dots,k\}$, the unitary interferometer $U_j$, acting on $m-j$ modes, depends on the measurement outcomes $p_1,\dots,p_j$. The adaptive measurement outcomes $p_1,\dots,p_k$ are used to drive the computation, whose final outcome is $s_1,\dots,s_{m-k}$.}
		\label{fig:adaptive}
	\end{center}
\end{figure*}

We write $\Pr_{m,n}[.|\bm t]$ the probability distribution of the outputs over $\Phi_{m,n}$ of the unitary interferometer $U$ acting on an input $\ket{\bm t}$. With the previous notations we obtain, for all $p,q\in\mathbb{N}$, all $\bm s\in\Phi_{m,p}$ and all $\bm t\in\Phi_{m,q}$,
\be
\text{Pr}_{m,n}[\bm s|\bm t]=\frac{|\Per(U_{\bm s,\bm t})|^2}{\bm s!\bm t!}\delta_{pq},
\label{proba}
\ee
where $\delta_{pq}$ is the Kronecker symbol.

In what follows, we fix the input state $\ket{\bm t}=\ket{\bm1^n\bm0^{m-n}}$, with single-photon states in the first $n$ modes and vacuum states in all other modes, where the superscript indicates the size of the string $(0,\dots,0)$ or $(1,\dots,1)$ when there is a possible ambiguity. We consider the case of linear optical quantum computing with adaptive photon-number measurements, which we refer to as adaptive linear optics (see Fig.~\ref{fig:adaptive}). We denote by $k\in\{0,\dots,m\}$ the number of single-mode adaptive measurements. Without loss of generality, we assume that the first $k$ modes are measured adaptively throughout the computation and we write $\bm p=(p_1,\dots,p_k)$ the adaptive measurement outcomes. For $r\in\mathbb N$ and $\bm p\in\Phi_{k,r}$, let us define
\be
U^{\bm p}:=\left[\mathbb 1_k\oplus U_k(p_1,\dots,p_k)\right]\dots\left[\mathbb 1_1\oplus U_1(p_1)\right]U_0,
\label{Up}
\ee
where $\mathbb 1_j$ is the identity matrix of size $j$ and where the unitary matrices $U_j$ depend on the measurement outcomes $p_1,\dots,p_j$ for all $j\in\{1,\dots,k\}$.
An adaptive interferometer $\mathcal U$ over $m$ modes with $n$ input photons and $k$ adaptive measurements is then represented as a family of nonadaptive unitary interferometers 
\be
\label{eq:adapUnit}
\mathcal U := \left\{U^{\bm p}\;|\;\bm p\in\Phi_{k,r}, \,0\le r\le n\right\},
\ee
for each possible adaptive measurement outcome $\bm p$. The matrix $U^{\bm p}$ describes the interferometer in Fig.~\ref{fig:adaptive}, when the adaptive measurement outcome $\bm p=(p_1,\dots,p_k)\in\Phi_{k,r}$ has been obtained, where $r$ is the number of photons detected during the adaptive measurements. 
In this case, the output state is a pure state which reads:
\be
\Tr_k\left[(\ket{\bm p}\!\bra{\bm p}\otimes\mathbb 1_{m-k})\hat U^{\bm p}\ket{\bm t}\!\bra{\bm t}\hat U^{\bm p\dag}\right],
\label{outputALO}
\ee
where the partial trace is over the first $k$ modes and where $\ket{\bm p}$ denotes the $k$-mode Fock state $\ket{p_1\dots p_k}$. At the end of the computation, all the remaining $m-k$ modes are measured with photon-number detection, yielding the final outcome $\bm s=(s_1,\dots,s_{m-k})\in\Phi_{m-k,n-r}$.


\subsection{Quantum probability and overlap estimation}

We now detail how to perform probability estimation or overlap estimation with an adaptive linear optical interferometer as described in the previous section.

For estimating an output probability with a quantum circuit, one may run the circuit $O(\poly m)$ times, obtaining classical outcomes, for which the frequency gives a polynomially precise additive estimate of the probability which can be computed efficiently. In the case of a circuit with adaptive measurements, one only looks at the final measurement outcomes and the same holds for adaptive linear optical computations.

For estimating the overlap of output states of two unitary quantum circuits, one may run both circuits in parallel and compare their quantum output states, for example with the swap test~\cite{buhrman2001quantum}. Doing so a polynomial number of times provides a polynomially precise estimate of the overlap. Alternatively, writing $C_1$ and $C_2$ the unitary circuits, one may build the circuit $C_1C_2^\dag$ and project the output quantum state onto the input state.

In the case of circuits with adaptive measurements, there is a different output state for each adaptive measurement outcome. In particular, if the number of possible adaptive measurement outcomes is exponential in the size of the computation, then the probability distribution for these outcomes has to be concentrated on a polynomial number of events for the quantum overlap estimation to be efficient. This is because in order to compute a polynomially precise estimate of the overlap, say, $|\braket{\phi|\psi}|^2$, the states $\ket\phi$ and $\ket\psi$, both corresponding to specific adaptive measurement results, have to be obtained a polynomial number of times.

For adaptive linear optics over $m$ modes with $n$ input photons and $k$ adaptive measurements, the number of possible adaptive measurement outcomes is given by
\be
\ba
\sum_{r=0}^n{|\Phi_{k,r}|}&=\sum_{r=0}^n{\binom{k+r-1}r}\\
&=\binom{n+k}n,
\ea
\ee
where the sum is over the total number of photons detected at the stage of the adaptive measurements. Hence, for overlap estimation to be efficient, either the probability distribution for the adaptive measurements outcomes is concentrated on a polynomial number of outcomes, or $\binom{n+k}{n}=O(\poly m)$, which is the case for example when $n=O(1)$ and $k=O(m)$, or $n=O(\log m)$ and $k=O(\log m)$, or $n=O(m)$ and $k=O(1)$.
In what follows, we do not assume concentration of the adaptive measurement outcome probability distribution and consider general interferometers with adaptive measurements. In this setting, the quantum efficient regime for overlap estimation thus corresponds to $\binom{n+k}{n}=O(\poly m)$.

Let $\mathcal U=\{U^{\bm r}|\bm r\in\Phi_{k,r},\,0\le r\le n\}$ be an adaptive linear interferometer with $n$ input photons and $k$ adaptive measurements and let $\ket\phi$ and $\ket\psi$ be two output states. Let $\bm p$ and $\bm q$ denote the outcomes of the adaptive measurements for $\ket\phi$ and $\ket\psi$, respectively, so that $U^{\bm p}$ is the interferometer for $\ket\phi$ and $U^{\bm q}$ is the interferometer for $\ket\psi$, with input Fock state $\ket{\bm t}$. With Eq.~\eqref{outputALO} we obtain
\be
\ba
\,&|\braket{\phi|\psi}|^2\\
&\!=\Tr\left[\Tr_k[(\ket{\bm p}\!\bra{\bm p}\otimes\mathbb 1_{m-k})\hat U^{\bm p}\ket{\bm t}\!\bra{\bm t}\hat U^{\bm p\dag}]\ket\psi\!\bra\psi\right]\\
&\!=\Tr\left[(\ket{\bm p}\!\bra{\bm p}\otimes\mathbb 1_{m-k})\hat U^{\bm p}\ket{\bm t}\!\bra{\bm t}\hat U^{\bm p\dag}(\mathbb 1_k\otimes\ket\psi\!\bra\psi)\right]\\
&\!=\Tr\left[\hat U^{\bm p}\ket{\bm t}\!\bra{\bm t}\hat U^{\bm p\dag}(\ket{\bm p}\!\bra{\bm p}\otimes\ket\psi\!\bra\psi)\right]\\
&\!=\Tr\left[\ket{\bm t}\!\bra{\bm t}\hat U^{\bm p\dag}\left(\ket{\bm p}\!\bra{\bm p}\otimes\ket\psi\!\bra\psi\right)\hat U^{\bm p}\right].
\ea
\label{overlapALO}
\ee
Because of the conservation of the total number of photons, the overlap between the states $\ket\phi$ and $\ket\psi$ is zero if $|\bm p|\neq|\bm q|$.
Otherwise, it can be estimated using a polynomial number of copies of the state $\ket\psi$ as follows: send the input $\ket{\bm p}\otimes\ket\psi$ into the interferometer with unitary matrix $U^{\bm p\dag}$ and measure the photon number in each output mode; record the value $1$ if the measurement pattern matches the Fock state $\ket{\bm t}$ and the value $0$ otherwise. Then, the mean of the obtained values yields a polynomially precise estimate of the overlap $|\braket{\phi|\psi}|^2$ by Eq.~\eqref{overlapALO} and Hoeffding inequality. A similar procedure can be followed when $\ket\phi$ and $\ket\psi$ are output states of two different adaptive interferometers.

Note that this overlap estimation requires the preparation of the Fock state $\ket{\bm p}$ but no adaptive measurements. By symmetry, one could estimate the overlap alternatively using a polynomial number of copies of the state $\ket\phi$ and preparing the Fock state $\ket{\bm q}$, or by mixing copies of $\ket\phi$ and $\ket\psi$. In practice, one may run the adaptive interferometer $\mathcal U$ and apply the above procedure for $\ket\phi$ or $\ket\psi$ on the fly, depending on whether the adaptive measurement outcome obtained is equal to $\bm q$ or $\bm p$ (see Algorithm~\ref{algo:kernelEstimation} for overlap estimation including this state preparation step).

These tasks of probability or overlap estimation may be performed in particular as subroutines for machine learning algorithms, as we detail in the next section.

\begin{algorithm}[ht]
 \textbf{Input:} Adaptive interferometer $\mathcal U=\{U^{\bm r}|\bm r\in\Phi_{k,r},\,0\le r\le n\}$, adaptive measurement outcomes $\bm p,\bm q$.
 
 \textbf{Parameters:}  Input $\ket{\bm t}$ and number of shots $T$.
 
 Set the state preparation counter $c_{sp}=0$.
 
 Set the overlap counter $c_{over}=0$.
 
 \While{$c_{sp}<T$}{Run $\mathcal U$ on input $\ket{\bm t}$, obtaining adaptive measurement outcome $\bm r$ and output state $\ket\chi$.
 
 \eIf{$\bm r=\bm p$}{$c_{sp}\rightarrow c_{sp}+1$.
  
  Run $U^{\bm q\dag}$ on input $\ket{\bm q}\otimes\ket\chi$.
  
  Measure in photon number basis.
  
  Record measurement outcome $\bm s$.
  
  \eIf{$\bm s=\bm t$}{$c_{over}\rightarrow c_{over}+1$,}{$c_{over}\rightarrow c_{over}$.}
 }
 
 \eIf{$\bm r=\bm q$}{$c_{sp}\rightarrow c_{sp}+1$.
  
  Run $U^{\bm p\dag}$ on input $\ket{\bm p}\otimes\ket\chi$.
  
  Measure in photon number basis.
  
  Record measurement outcome $\bm s$.
  
  \eIf{$\bm s=\bm t$}{$c_{over}\rightarrow c_{over}+1$,}{$c_{over}\rightarrow c_{over}$.}
 }{}
 }
\textbf{return} $c_{over}/T$.
 \caption{Overlap estimation for adaptive linear optics}
  \label{algo:kernelEstimation}
\end{algorithm}
%


\subsection{Support vector machine with adaptive linear optics}
\label{sec:SVM}


We consider a training dataset $\mathcal{T}$, with $\vert \mathcal{T} \vert$ points of the form $\{(\vv x_1,y_1),\dots,(\vv x_{\vert \mathcal{T} \vert},y_{\vert \mathcal{T} \vert})\}$, where $\vv x_l \in \mathbb{R}^d$, $y_l \in C$, $\forall l\in \{1\dots,{\vert \mathcal{T} \vert}\}$ and where $C = \{-1,1\}$ in the case of binary classification. 

We also consider a feature map that is given by a subset of the family introduced in Eq.~\eqref{eq:adapUnit}: $\mathcal{U}_\phi = \{U_{l}^{\bm p_{l}}\}_l$, in which for each classical data $x_l$ there is a unitary operator $U_{l}^{\bm p_{l}}$.

\subsubsection{Support vector machine with quantum kernel methods}

The goal of a support vector machine \cite{10.5555/211359, Burges1998, 10.5555/2517747} is to find the maximum-margin hyperplane that divides the set of points by their $y$ values. Such an hyperplane is defined by a vector $\vv w \in \mathbb{R}^d$ and $b\in \mathbb{R}$. With the hard-margin condition we have,
\begin{equation}\label{eq:SVMHyperplane}
   y_l(\vv w. \vv x_l + b) \geq 1, \quad \forall l \in \{1\dots,{\vert \mathcal{T} \vert}\}.
\end{equation}
In this case, a decision function over a new point $\vv x$ can be constructed directly from the hyperplane as follows:
\begin{equation}\label{eq:decisionFunction}
   f(\vv x) = \text{sign}(\vv w.\vv x + b).
\end{equation}
When we have access to a feature map $\phi$, the vector $\vv{w}$ can be decomposed as 
\begin{equation}\label{eq:classVecFeatSpace}
   \vv w =\sum_{l=1}^{|\mathcal{T}|} \alpha_l y_l \phi(x_l),
\end{equation}
for some $\alpha_l\in\mathbb C$, and the decision function becomes
\begin{equation}\label{eq:decisionFunction2}
   f(\vv x) = \text{sign}\left(\sum_{l=1}^{\vert \mathcal{T} \vert} \alpha_l y_l \kappa(\vv x_l,\vv x) + b\right),
\end{equation}
where the kernel $\kappa(\vv x_l,\vv x)$ has an interpretation in terms of the feature map $\phi$ as given in Eq.~\eqref{eq:fMapKernel}.

In order to find the best separating hyperplane, one may use a quadratic programming solver.
In section~\ref{sec:exp}, we consider the case where the programming solver is composed of a hybrid quantum-classical algorithm whereas in section~\ref{sec:imp}, such solver will be external to the quantum device.
The optimisation program that needs to be solved is:
\begin{equation}  \label{eq:lagSVMOpt}
	\begin{split}
	    \underset{\alpha_l}{\text{maximize}} \;& \sum_{l=1}^{\vert \mathcal{T} \vert} \alpha_l - \frac{1}{2} \sum_{l,l'} \alpha_l y_l\alpha_{l'} y_{l'} \kappa(\vv x_l,\vv x_{l'}) \\
		\text{subject to}&\;\sum_{l=1}^{\vert \mathcal{T} \vert} \alpha_l y_l = 0 \\
		&\; 0 \leq \alpha_l \leq 1/2n\lambda , \quad \forall l\in \{1\dots,\vert \mathcal{T} \vert\},
	\end{split}
\end{equation}
where $\lambda$ is a parameter that is introduced for the soft-margin condition. By solving such problem, one obtains the optimal coefficients $\{\alpha_l^*\}$ and $b^*$ of the hyperplane (see Eq.~\eqref{eq:classVecFeatSpace}). This formulation is for the quantum case , but similar methods can be applied in the classical case \cite{Burges1998}.


\subsubsection{Explicit method: probability estimation}\label{sec:exp}

The explicit method corresponds to the case where the prediction---assigning a classification label to the data---is obtained from the probability distribution when using a quantum device.

\begin{algorithm}[h]
 \textbf{Input:} Training dataset $\mathcal{T}$, feature maps $\mathcal{U}_\phi$, optimisation routine with cost function $J(\vv \theta)$.
 
 \textbf{Parameters:} Input $\ket{\bm t}$, number of shots $T$, initial parameters $\vv \theta_0$.
 
 Set $\vv \theta = \vv \theta_0$.
 
 \While{$J(\vv \theta)$ not converged}{
 
  \For{$l=1$ to ${\vert \mathcal{T} \vert}$}{
  
    Set $c_{p_{l}} = 0$. 
    
    Set $c_y = 0$ $\forall y \in C$.
    
    \While{$c_{p_{l}} < T$}{Run $BS(\vv \theta) U_{l}^{\bm p_{l}}$ on $\ket{\bm t}$, obtaining adaptive measurement outcome $\bm r$ and outcome label $\bm y$.

        \eIf{$\bm r = \bm p_{l}$}
            {$c_{p_{l}}\rightarrow c_{p_{l}}+1$.
        
            $c_y\rightarrow c_y+1$.
         }{}
    \textbf{end while}
    }
    Compute $\Pr(y\vert \bm p_{l}) = c_y/T$.
    
  \textbf{end for}
  }
  Compute $J(\vv \theta)$ and update $\vv \theta$.
  
 \textbf{end while}
 }

\textbf{return} Value $J(\vv \theta^*)$ and final $\vv \theta^*$ .
 \caption{Explicit method: training phase}
  \label{algo:expTrainPh}
\end{algorithm}

\begin{algorithm}[ht]
 \textbf{Input:} Unlabeled data $x$, feature map $\mathcal{U}_{\phi(x)}$, optimal parameters $\vv \theta^*$.
 
 \textbf{Parameters:} Input $\ket{\bm t}$, number of shots $T_x$.
 
    Set $c_{p_{x}} = 0$. 
    
    Set $c_y = 0$ $\forall y \in C$.
    
    \While{$c_{p_{x}} < T_x$}{Run $BS(\vv \theta^*) \mathcal{U}_{\phi(x)}$ on $\ket{\bm t}$, obtaining adaptive measurement outcome $\bm r$ and outcome label $\bm y$.
    
        \eIf{$\bm r = \bm p_{x}$}{$c_{p_{x}}\rightarrow c_{p_{x}}+1$.
        
            $c_y\rightarrow c_y+1$.
         }{}
    \textbf{end while}
    }
    Compute $\Pr(y\vert \bm p_{x}) = c_y/T_x$.
    
\textbf{return} $\text{argmax}_y \{\Pr(y\vert \bm p_{x})\}$ .
 \caption{Explicit method: prediction phase}
  \label{algo:expPredPh}
\end{algorithm}

Hereafter, we consider using adaptive linear optics for the feature map and use its output state as an input state of a second quantum device for the support vector machine \cite{havlivcek2019supervised}. The choice of the feature map ultimately depends on the architecture at hand: one may use, e.g., reconfigurable interferometers~\cite{politi2009integrated} in the linear optics picture. We show that Boson Sampling can be used to realize a support vector machine algorithm by using a hybrid variational quantum-classical for the training phase \cite{PhysRevA.98.032309, farhi2018classification, PhysRevResearch.1.033063}. 
We write the Boson Sampling interferometer operation $BS(\vv \theta)$, where $\vv \theta$ are the parameters of the beam splitters and phase shifters of the interferometer, since any interferometer can be implemented efficiently with only phase shifters and beam splitters~\cite{reck1994experimental}. In order to make a prediction we need to bin the outcomes. For a given function $g:\Phi_{m,n}\rightarrow \{-1,+1\}$, we can write the following observable for the binning:
\be \label{eq:qmeasBSSVM}
    \mathcal{N} = \sum_{s \in \Phi_{m,n}} g(s) \ket{\bm s} \bra{\bm s}.
\ee
The observable $\mathcal{N}$ can also be decomposed in term of projectors as $\mathcal{N} = \Pi_{+1} - \Pi_{-1}$, where $\Pi_y = (\mathbb1+y \mathcal{N})/2$.

For a given data point $\vv x$ with the associated operation $\hat U_x^{\bm p_{x}}$, the probability to obtain the outcome $y$ is:
\be
\ba
& \Pr(y\vert \bm p_{x}) = \Tr \left[\Pi_y BS(\vv \theta) \hat U_{x}^{\bm p_{x}} \ket{\bm t} \bra{\bm t} \hat U_{x}^{\bm p_{x}\dagger} BS(\vv \theta)^\dagger \right]\\
&\quad=\frac12\left(1+y \bra{\bm t} \hat U_{x}^{\bm p_{x}\dagger} BS(\vv \theta)^\dagger \mathcal{N} BS(\vv \theta) \hat U_{x}^{\bm p_{x}} \ket{\bm t}\right).
\ea
\label{eq:probSVM}
\ee
In the quantum circuit model, it has been explicitly shown in \cite{havlivcek2019supervised} that the equivalent of Eq.~\eqref{eq:probSVM} in the circuit picture is related to the decision function of a support vector machine. This proof relies on decomposing the variational quantum circuit in the Pauli basis. Moreover, in \cite{PhysRevA.97.062329}, the authors provide a quantum circuit that simulates Boson Sampling interferometers with input photons, using the quantum Schur transform. Hence, the unitary $BS(\vv \theta)$ can be decomposed in the Pauli basis (for qudits), thus providing the same relation for Eq.~\eqref{eq:probSVM} to the decision function of a support vector machine.

Experimentally, the probability $\Pr(y\vert \bm p_{x})$ can be estimated using the following approximation

\be \label{eq:qmeasBSSVM2}
    \Pr(y\vert \bm p_{x}) \approx \frac{T_{y\vert x}}{T_x},
\ee
where $T_x$ is the number of times where the value $\bm p_{x}$ has been recorded and $T_{y\vert x}$ is the number of times where the value $y$ has been recorded after that the value $\bm p_{x}$ has been also recorded.

In order to train the Boson Sampling device, it is necessary to use a cost function $J(\vv \theta)$ and an optimisation routine. For the optimisation routine, standard variational optimisation methods can be used such as hybrid variational quantum-classical or classical algorithms \cite{PhysRevA.98.032309, farhi2018classification, PhysRevResearch.1.033063}. Different cost functions can be used such as the empirical risk function \cite{havlivcek2019supervised} or the square-loss function \cite{schuld2019quantum}. However, other type of cost function may not be suitable for this kind of problems, such as the ones usually used in gate synthesis or quantum state preparation \cite{arrazola2018machine,heya2018variational}.

The algorithm for the training phase is described in Algorithm~\ref{algo:expTrainPh} and the prediction phase is described in Algorithm~\ref{algo:expPredPh}.


\subsubsection{Implicit method: overlap estimation}\label{sec:imp}

\begin{algorithm}[h]
 \textbf{Input:} Dataset $\mathcal{T}$, Feature maps $\mathcal{U}$, quadratic programming solver.
 
 \textbf{Parameters:} Input $\ket{\bm t}$ and number of shots $T$.
 
 \For{$l=1$ to ${\vert \mathcal{T} \vert}$}{
     \For{$l'=1$ to ${\vert \mathcal{T} \vert}$}{
      Run $U_{l}^{\bm p_{l}\otimes T}$ on input $\ket{\bm t}^{\otimes T}$.
      
      Run Algo.~\ref{algo:kernelEstimation} with input $U_{l}^{\bm p_{l}},U_{l'}^{\bm p_{l'}}$.
      
      Store output at kernel matrix entry $K_{l,l'}=\kappa(\vv x_l,\vv x_{l'})$.

    \textbf{endfor}}
\textbf{endfor}}
Solve the optimisation problem in Eq.~\eqref{eq:lagSVMOpt} with $\mathcal{T}$ and $K$.

\textbf{return} $\{\alpha_l^*\}$ and $b^*$.
 \caption{Implicit method: training phase}
  \label{algo:impTrainPh}
\end{algorithm}

The implicit method corresponds to the case where the kernel is computed by a quantum device, while the rest of the machine learning algorithm is performed on a classical machine. In adaptive linear optics, the kernel can be estimated using Algorithm~\ref{algo:kernelEstimation}. It is possible to proceed to both the training and prediction phase using this hybrid quantum-classical method. The algorithm for the training phase is described in Algorithm~\ref{algo:impTrainPh} and the prediction phase is described in Algorithm~\ref{algo:impPredPh}.

\begin{algorithm}[h]
 \textbf{Input:} Dataset $\mathcal{T}$, unlabeled data $\vv x$, optimal parameters $\{\alpha_l^*\}$ and $b$
 
 \textbf{Parameters:} Input $\ket{\bm t}$ and number of shots $T$.
 
 Set the variable $f = b$.
 
 \For{$l=1$ to ${\vert \mathcal{T} \vert}$}{
  Run $U_{x}^{\bm p_{x}\otimes T}$ on input $\ket{\bm t}^{\otimes T}$.
  
  Run Algo.~\ref{algo:kernelEstimation} with parameters $U_{x}^{\bm p_{x}\otimes T}\ket{\bm t}^{\otimes T},U_{l}^{\bm p_{l}}$.
  
  $f \rightarrow f+ \alpha_l y_l \kappa(\vv x_l,\vv x)$.
  
\textbf{endfor}}

\textbf{return} Return sign$(f)$.
 \caption{Implicit method: prediction phase}
  \label{algo:impPredPh}
\end{algorithm}


\section{Classical simulation of adaptive linear optics}
\label{sec:simulation}

The complexity of probability estimation and overlap estimation of quantum computations has been well studied in the circuit model~\cite{pashayan2015estimating,bravyi2019classical}. In this section, we challenge the quantum algorithms introduced in the previous section by deriving classical algorithms for probability estimation and overlap estimation for adaptive linear optics over $m$ modes. We identify various complexity regimes for different numbers of input photons $n\le m$, different number of adaptive measurements $k\le m$, and different number of photons $r\le n$ detected during the adaptive measurements. To do so, we derive generic expressions for the output probabilities and for the overlap of output states of linear optical interferometers with adaptive measurements.


\subsection{Classical probability estimation}

Firstly, we obtain a classical algorithm for probability estimation of adaptive linear optics over $m$ modes with $n$ input photons and $k$ adaptive measurements.

We first consider the case $k=0$, i.e., Boson Sampling. The probability of an outcome $\bm s\in\Phi_{m,n}$ for a unitary interferometer $U$ given the input $\bm t=(\bm1^n,\bm0^{m-n})\in\Phi_{m,n}$ is given by Eq.~\eqref{proba}:
\be
\text{Pr}_{m,n}[\bm s]=\frac1{\bm s!}|\Per\,(U_{\bm s,\bm t})|^2,
\ee
where $U_{\bm s,\bm t}$ is the $n\times n$ matrix obtained from $U$ by repeating $s_i$ times its $i^{th}$ row for $i\in\{1,\dots,m\}$ and removing its $j^{th}$ column for $j=\{n+1,\dots,m\}$. When $|\bm s|\neq n$ however, the probability is $0$, since the input $\bm t$ has $n$ photons and the linear interferometer does not change the total number of photons.

The permanent of a square matrix of size $n$ can be computed exactly in time $O(n2^n)$, thanks to Ryser's formula~\cite{albert1978combinatorial}. However, polynomially precise estimates of the permanent of a square unitary matrix can be obtained in polynomial time in the size of the matrix using an algorithm due to Gurvits~\cite{gurvits2005complexity}, later generalised to matrices with repeated lines or columns~\cite{aaronson2012generalizing}, so probability estimation can be done classically efficiently, which was already noted in~\cite{Aaronson2013}.

We now turn to the case $k>0$---which to our knowledge has not been treated elsewhere---using the notations of section~\ref{sec:ALO1}. This case is a direct extension of the case $k=0$. With Eq.~\eqref{proba}, for $r\in\mathbb N$, $\bm p\in\Phi_{k,r}$ and $\bm s\in\Phi_{m-k,n-r}$, the probability of an total outcome $(\bm p,\bm s)\in\Phi_{m,n}$ (adaptive measurement and final outcome) is given by
\be
\text{Pr}^{\text{total}}_{m,n}[\bm p,\bm s]=\frac{1}{\bm p!\bm s!}\left|\Per\left(U^{\bm p}_{(\bm p,\bm s),\bm t}\right)\right|^2.
\ee
Let $r\in\{0,\dots,n\}$ and let $\bm s\in\Phi_{m-k,n-r}$. Then, the probability of obtaining the final outcome $\bm s$ after the adaptive measurements reads
\be
\ba
\text{Pr}_{m,n}^{\text{final}}[\bm s]&=\sum_{\bm p\in\Phi_{k,r}}{\text{Pr}^{\text{total}}_{m,n}[\bm p,\bm s]}\\
&=\frac1{\bm s!}\sum_{\bm p\in\Phi_{k,r}}{\frac{1}{\bm p!}\left|\Per\left(U^{\bm p}_{(\bm p,\bm s),\bm t}\right)\right|^2}.
\ea
\label{probak}
\ee
The sum is taken over the elements of $\Phi_{k,r}$, which has $\binom{k+r-1}r$ elements, where $r\le n$ is the total number of photons detected during the adaptive measurements. Each permanent in the sum can be estimated additively with polynomial precision in time $O(\poly m)$, using the generalised algorithm from~\cite{aaronson2012generalizing} (see Appendix~\ref{app:simproba} for details). Hence, probability estimation can be done classically efficiently whenever the sum has a polynomial number of terms. In particular, as long as both $k$ and $r$ are $O(\log m)$, the output probability can be estimated efficiently. The simulability regimes are summarised in Table~\ref{tab:probaest}. 

The universal quantum computing regime corresponds to $n=O(m)$ and $k=O(m)$~\cite{knill2001scheme}. The time complexity of the classical simulation is $O\left(\binom{k+r-1}r\poly m\right)$ and there is a possibility of subuniversal quantum advantage for probability estimation for $n=O(\log m)$ and $k=O(m)$, or $n=O(m)$ and $k=O(\log m)$. Moreover, the fraction $\frac rn$ of input photons detected during the adaptive measurements has to be sufficiently large to prevent efficient classical simulation.

\begin{table}
\centering
\setlength\tabcolsep{0.1pt}
\bgroup
\def\arraystretch{2}
\begin{tabular}{| c | c | c | c |}
\hline
\diagbox[height=7ex,width=5em]{$\quad\; r$}{$k\;\;$} & $O(1)$ & $O(\log m)$ & $O(m)$\\
\hline
$O(1)$ & \cellcolor{green!15} $\quad\quad\quad\quad\;$& \cellcolor{green!15} $\quad\quad\quad\quad\;$& \cellcolor{green!15} $\quad\quad\quad\quad\;$\\
\hline
$O(\log m)$ & \cellcolor{green!15} & \cellcolor{green!15} & \cellcolor{red!15} \\
\hline
$O(m)$ & \cellcolor{green!15} &  \cellcolor{red!15} & \cellcolor{red!15} \\
\hline
\end{tabular}
\egroup
\caption{Simulability regimes for probability estimation as a function of the parameters $r$ (the total number of photons detected during the adaptive measurements) and $k$ (the number of single-mode adaptive measurements). In green is the parameter regime for which the classical probability estimation is efficient, i.e., takes polynomial time in $m$, while in red is the regime where it is no longer efficient.}
\label{tab:probaest}
\end{table}
%


\subsection{Classical overlap estimation}

In this section, we obtain a classical algorithm for overlap estimation of output states of adaptive linear optical interferometers over $m$ modes with $n$ input photons and $k$ adaptive measurements.

Once again, we start with $k=0$. With Eq.~\eqref{transf}, the output state of an $m$-mode interferometer $U$ with input state $\bm t\in\Phi_{m,n}$ reads
\be
\ba
\ket\phi&=\sum_{\bm s\in\Phi_{m,n}}{\braket{\bm s|\hat U|\bm t}\ket{\bm s}}\\\
&=\sum_{\bm s\in\Phi_{m,n}}{\frac{\Per\,(U_{\bm s,\bm t})}{\sqrt{\bm s!\bm t!}}\ket{\bm s}},
\ea
\ee
where $U_{\bm s,\bm t}$ is the $n\times n$ matrix obtained from $U$ by repeating $s_i$ times its $i^{th}$ row for $i\in\{1,\dots,m\}$ and repeating $t_j$ times its $j^{th}$ row for $j\in\{1,\dots,m\}$. The composition of two interferometers is another interferometer which unitary representation is the product of the unitary representations of the composed interferometers.
Hence, the inner product of the output states $\ket\phi$ and $\ket\psi$ of two $m$-mode interferometers $U$ and $V$ with the same input state $\bm t\in\Phi_{m,n}$, is equal to the matrix element $\bm t,\bm t$ of $\hat U^\dag\hat V$:
\be
\ba
\braket{\phi|\psi}&=\sum_{\bm u,\bm v\in\Phi_{m,n}}{\braket{\bm t|\hat U^\dag|\bm u}\braket{\bm v|\hat V|\bm t}\braket{\bm u|\bm v}}\\
&=\sum_{\bm s\in\Phi_{m,n}}{\braket{\bm t|\hat U^\dag|\bm s}\braket{\bm s|\hat V|\bm t}}\\
&=\braket{\bm t|\hat U^\dag\hat V|\bm t}\\
&=\frac{\Per\left[(U^\dag V)_{\bm t,\bm t}\right]}{\bm t!},
\ea
\ee
where we used in the third line $\bm t\in\Phi_{m,n}$ and the fact that $\hat U^\dag\hat V$ conserves the space $\Phi_{m,n}$. With the input $\bm t=(\bm1^n,\bm0^{m-n})$ with $n$ photons in $m$ modes, this reduces to
\be
\ba
\braket{\phi|\psi}=\Per\left[(U^\dag V)_n\right],
\ea
\ee
where $(U^\dag V)_n$ is the $n\times n$ top left submatrix of $U^\dag V$. Hence, the inner product and the overlap may be approximated to a polynomial precision efficiently, since this is the case for the permanent~\cite{gurvits2005complexity,aaronson2012generalizing}, as detailed in the previous section.

We now consider the case $k>0$. Let $r\in\mathbb N$ and let $\bm p\in\Phi_{k,r}$. With Eq.~\eqref{transf} and Eq.~\eqref{outputALO}, writing $\text{Pr}^{\text{adap}}_{m,n}[\bm p]$ the probability of obtaining the adaptive measurement outcome $\bm p$, the output state of the interferometer $U^{\bm p}$ with $k$ adaptive measurements with input $\bm t=(\bm1^n,\bm0^{m-n})$ in Fig.~\ref{fig:adaptive}, when the adaptive measurement outcome $\bm p$ is obtained, reads
\be
\frac1{\sqrt{\text{Pr}^{\text{adap}}_{m,n}[\bm p]}}\ket{\psi_{\bm p}},
\ee
where
\be
\ket{\psi_{\bm p}}:=\sum_{\bm s\in\Phi_{m-k,n-r}}{\frac{\Per\left(U^{\bm p}_{(\bm p,\bm s),\bm t}\right)}{\sqrt{\bm p!\bm s!}}\ket{\bm s}},
\ee
and where $\text{Pr}^{\text{adap}}_{m,n}[\bm p]=\braket{\psi_{\bm p}|\psi_{\bm p}}$. The inner product of two (not normalised) output states $\ket{\psi_{\bm p}}$ and $\ket{\psi_{\bm q}}$ of $m$-mode interferometers $U^{\bm p}$ and $V^{\bm q}$ with $k$ adaptive measurements is zero if $|\bm p|\neq|\bm q|$, and when $|\bm p|=|\bm q|=r$, it is thus given by
\be
\ba
\,&\braket{\psi_{\bm p}|\psi_{\bm q}}=\frac1{\sqrt{\bm p!\bm q!}}\\
&\times\!\!\!\sum_{\bm s\in\Phi_{m-k,n-r}}{\frac1{\bm s!}\Per\left(U^{\bm p\dag}_{\bm t,(\bm p,\bm s)}\right)\Per\left(V^{\bm q}_{(\bm q,\bm s),\bm t}\right)}.
\ea
\label{overlapS}
\ee
\begin{table}[t]
\centering
\setlength\tabcolsep{0.1pt}
\bgroup
\def\arraystretch{2}
\begin{tabular}{| c | c | c | c |}
\hline
\diagbox[height=7ex,width=5em]{$\quad\; n$}{$k\;\;$} & $O(1)$ & $O(\log m)$ & $O(m)$\\
\hline
$O(1)$ & \cellcolor{green!15}$\quad\quad\quad\quad\quad$ & \cellcolor{green!15}$\quad\quad\quad\quad\quad$ & \cellcolor{green!15}$\quad\quad\quad\quad\quad$\\
\hline
$O(\log m)$ & \cellcolor{green!15}& \cellcolor{green!15}& \cellcolor{green!15}$\cdots$\\
\hline
$O(m)$ & \cellcolor{red!15}& \cellcolor{red!15}$\cdots$& \cellcolor{red!15}$\cdots$\\
\hline
\end{tabular}
\egroup
\caption{Simulability regimes for classical overlap estimation as a function of the parameters $n$ and $k$. In green is the parameter regime for which the classical overlap estimation may be efficient (i.e., may take polynomial time in $m$) depending on the value of $r$, and in red is the regime where the classical algorithm is no longer efficient. The cells containing a symbol ``$\cdots$'' correspond to parameter regimes for which the quantum algorithm for estimating the overlap is no longer efficient.}
\label{tab:overlapest}
\end{table}

This expression is a sum of $|\Phi_{m-k,n-r}|$ terms, which is exponential in $m$ whenever $n$ is not constant.
Using properties of the permanent~\cite{percus2012combinatorial}, we show that the expression in Eq.~\eqref{overlapS} may actually be rewritten as a sum over fewer terms, which constitutes our main technical result:

\begin{lem}\label{lem:overlapS}
Let $r\in\mathbb N$. The inner product of two (not normalised) output states $\ket{\psi_{\bm p}}$ and $\ket{\psi_{\bm q}}$ of $m$-mode interferometers $U^{\bm p}$ and $V^{\bm q}$ with adaptive measurements outcome $\bm p,\bm q\in\Phi_{k,r}$ is given by
\be
\ba
\,&\braket{\psi_{\bm p}|\psi_{\bm q}}=\frac1{\sqrt{\bm p!\bm q!}}\\
&\times\sum_{\substack{\bm i,\bm j\in\{0,1\}^n\\|\bm i|=|\bm j|=r}}{\Per\left(A^{\bm i}\right)\Per\left(B^{\bm j}\right)\Per\left(C^{\bm i,\bm j}\right)},
\ea
\ee
where for all $\bm i,\bm j\in\{0,1\}^n$ such that $|\bm i|=|\bm j|=r$,
\be
A^{\bm i}=U^{\bm p\dag}_{(\bm i,\bm0^{m-n}),(\bm p,\bm0^{m-k})}
\ee
is an $r\times r$ matrix which can be obtained efficiently from $U^{\bm p}$,
\be
B^{\bm j}=V^{\bm q}_{(\bm q,\bm0^{m-k}),(\bm j,\bm0^{m-n})}
\ee
is an $r\times r$ matrix which can be obtained efficiently from $V^{\bm q}$, and
\be
C^{\bm i,\bm j}\!\!=\!U^{\bm p\dag}_{(\bm1^n-\bm i,\bm0^{m-n}),(\bm0^k,\bm1^{m-k})}V^{\bm q}_{(\bm0^k,\bm1^{m-k}),(\bm1^n-\bm j,\bm0^{m-n})}
\ee
is an $(n-r)\times(n-r)$ matrix which can be obtained efficiently from $U^{\bm p}$ and $V^{\bm q}$.
\end{lem}

\noindent We give a detailed proof in Appendix~\ref{app:prooflemma2}. By Lemma~\ref{lem:overlapS}, the inner product is expressed as the modulus squared of a sum over $\binom nr^2$ products of three permanents, of square matrices of sizes $|\bm p|=r$, $|\bm q|=r$ and $(n-r)$, respectively. In the worst case, when $r=n/2$, the sum has at most $O(4^n)$ terms, up to a polynomial factor in $n$. In particular, when $n=O(\log m)$ or $r=O(1)$, the inner product reduces to a sum of a polynomial number of terms, which can all be computed in time $O(\poly m)$ with Ryser's formula~\cite{albert1978combinatorial}. An interesting fact is that the cost of computing the inner product does not depend explicitely on the number $k>0$ of adaptive measurements. However, it does depend explicitly on $r$ the total number of photons detected during the adaptive measurements, and a larger number of adaptive measurements $k$ favors the detection of a larger number of photons $r$.

\begin{table}[t]
\centering
\setlength\tabcolsep{0.1pt}
\bgroup
\def\arraystretch{2}
\begin{tabular}{| c | c | c | c |}
\hline
\diagbox[height=7ex,width=5em]{$\quad\; n$}{$r\;\;$} & $O(1)$ & $O(\log n)$ & $O(n)$\\
\hline
$O(1)$ & \cellcolor{green!15}$\quad\quad\quad\quad\quad$ & \cellcolor{green!15}$\quad\quad\quad\quad\quad$ & \cellcolor{green!15}$\quad\quad\quad\quad\quad$\\
\hline
$O(\log m)$ & \cellcolor{green!15}  & \cellcolor{green!15}  & \cellcolor{green!15} \\
\hline
$O(m)$ & \cellcolor{green!15}  & \cellcolor{red!15} & \cellcolor{red!15} \\
\hline
\end{tabular}
\egroup
\caption{Simulability regimes for classical overlap estimation as a function of the parameters $n$ (total number of input photons) and $r$ (total number of photons detected in the adaptive measurements). The regimes are obtained from the running time of the classical algorithm using Stirling's equivalent $n!\sim\sqrt{2\pi n}(\frac ne)^n$.}
\label{tab:overlapest2}
\end{table}

The overlap of normalised ouput states is given by
\be
\frac{|\braket{\psi_{\bm p}|\psi_{\bm q}}|^2}{\braket{\psi_{\bm p}|\psi_{\bm p}}\braket{\psi_{\bm q}|\psi_{\bm q}}},
\ee
which may also be computed efficiently when $n=O(\log m)$. In this case, the classical algorithm for overlap estimation simply computes the above expression, using Lemma~\ref{lem:overlapS} for each of the inner products. The running time of this classical algorithm thus is $O(\binom nr^2\poly m)$ and its efficiency is summarised as a function of $n$ and $k$ in Table~\ref{tab:overlapest} and as a function of $n$ and $r$ in Table~\ref{tab:overlapest2}.

Since the quantum efficient regime corresponds to $\binom{k+n}n=O(\poly m)$ there is a possibility of quantum advantage for overlap estimation when $k=O(1)$ and $n=O(m)$. However, like for probability estimation, the fraction $\frac rn$ of input photons detected during the adaptive measurements has to be sufficiently large to prevent efficient classical simulation. In this case, when the number of adaptive measurements satisfies $k=O(1)$, the interferometers used must concentrate many photons on the adaptive measurements.


\section{Conclusion}
\label{sec:conclusion}

In this work, we have given a roadmap for performing quantum variational classification and quantum kernel estimation using adaptive linear optical interferometers. 
We have investigated the transition of classical simulation between Boson Sampling~\cite{Aaronson2013} and the Knill--Laflamme--Milburn scheme for universal quantum computing~\cite{knill2001scheme}, in terms of the number of adaptive measurements performed. In particular, we have derived classical algorithms for simulating the quantum computational subroutines involved: output probability estimation and output state overlap estimation. 

In the case of probability estimation, the possible regimes for quantum advantage are incompatible with near-term implementations: both the number of adaptive measurements $k$ and the number of input photons $n$ must be greater than $\log m$, where $m$ is the number of modes. On the other hand, for overlap estimation, there is a possibility of near-term computational advantage using adaptive linear optics with a single adaptive measurement, requiring the preparation of photon number states. Note that the interferometer should be concentrating many photons $r$ at the stage of the adaptive measurements in order to obtain overlaps that are possibly hard to estimate. Using more adaptive measurements does not increase the complexity (apart from polynomial factors in $m$), but may ease the detection of a larger number of photons during the adaptive measurements.

Our results suggest regimes where quantum advantage for machine learning with adaptive linear optics is possible, in the parameter regimes where our classical simulation algorithms fail to be efficient: it is an interesting open question whether better classical simulation algorithms for the computational subroutines involved can be found, or even classical algorithms solving directly the machine learning problems efficiently. In any case, our results restrict the parameter regimes for which such a quantum advantage may be possible: support vector machine with adaptive linear optics has to take place either in a bunching regime---concentrating many photons in the adaptive measurements---or using a number of adaptive measurements scaling with the size of the problem. In both cases, this imposes strong experimental requirements. 

Our results have identified a sweet spot for quantum kernel estimation with adaptive linear optics using a single adaptive measurement, which would be interesting to demonstrate experimentally. In the longer term, variational classification with adaptive linear optics could also be interesting since it may enable quantum advantage in a regime where the quantum algorithm for overlap estimation is no longer efficient.

As previously mentioned, it is a pressing question whether more efficient classical algorithms may be derived. In practical settings, taking into account photon losses could help providing more efficient classical simulation algorithms~\cite{garcia2019simulating}. We leave these considerations for future work.


\section*{Acknowledgements}

AS has been supported by a KIAS individual grant (CG070301) at Korea Institute for Advanced Study.
DM acknowledges support from the ANR through project ANR-17-CE24-0035 VanQuTe.


\bibliographystyle{linksen}
\bibliography{bibliography}


\onecolumn\newpage
\appendix


\begin{center}
    {\huge Appendix}
\end{center}

\section{Classical probability estimation for Shor's period-finding algorithm}
\label{app:Shor}

In this section, we detail the classical probability estimation for Shor's period-finding algorithm~\cite{shor1994algorithms}.

\medskip

If $N$ is an $n$ bits integer to factor, the period-finding subroutine measures the output state
\be
\frac1N\sum_x\sum_ye^{\frac{2i\pi xy}N}\ket y\ket{f(x)}
\ee
in the computational basis, where $f$ is a periodic function over $\{0,\dots,N-1\}$ which can be evaluated efficiently. The probability of obtaining an outcome $y_0,f(x_0)$ is given by
\be
\Pr\,[y_0,f(x_0)]=\left|\frac1N\sum_{f(x)=f(x_0)}e^{\frac{2i\pi xy_0}N}\right|^2.
\ee
Now let
\be
g_{x_0,y_0}:x\mapsto\begin{cases}e^{\frac{2i\pi xy_0}N}\text{ if } f(x)=f(x_0),\\0\text{ otherwise.}\end{cases}
\ee
The function $g_{x_0,y_0}$ can be evaluated efficiently and we have
\be
\Pr\,[y_0,f(x_0)]=\left|\underset{x\leftarrow N}{\mathbb E}[g_{x_0,y_0}(x)]\right|^2,
\ee
where $\underset{x\leftarrow N}{\mathbb E}$ denotes the expected value for $x$ drawn uniformly randomly from $\{0,\dots,N-1\}$. By virtue of Hoeffding inequality, this quantity may be estimated efficiently (in $n$, the number of bits of $N$) classically by sampling uniformly a polynomial number of values in $\{0,\dots,N-1\}$ and computing the modulus squared of the mean of $g_{x_0,y_0}$ over these values.

However, note that---the decision version of---probability estimation of quantum circuits is a \textsf{BQP}-complete problem almost by definition, since given a polynomially precise estimate of the probability of acceptance of an input $x$ to a quantum circuit, one may determine whether it is accepted or rejected by the circuit. In particular, unless factoring is in \textsf{P}, probability estimation for the quantum circuit corresponding to Shor's algorithm \textit{as a whole} is hard for classical computers, and weak simulation of the period-finding subroutine is also hard, since in Shor's algorithm the output samples from the period-finding subroutine are used for a different classical computation than probability estimation, namely obtaining promising candidates for the period.


\section{Efficiency of classical output probability estimation}
\label{app:simproba}

For an interferometer $U$ over $m$ modes with $n$ input photons and $k$ adaptive measurements with outcomes $\bm p\in\Phi_{k,r}$, the expression of the probability of an outcome $\bm s$ obtained in the main text reads:
\be
\text{Pr}_{m,n}^{\text{final}}[\bm s]=\frac1{\bm s!}\sum_{\bm p\in\Phi_{k,r}}{\frac{1}{\bm p!}\left|\Per\left(U^{\bm p}_{(\bm p,\bm s),\bm t}\right)\right|^2},
\label{probaappfinal}
\ee
where $\bm t=(\bm1^n,\bm0^{m-n})$. This is a sum over $\left|\Phi_{k,r}\right|=\binom{k+r-1}r$ moduli squared of permanents of square matrices of size $n\le m$. Permanents can be approximated efficiently using a simple randomized algorithm due to Gurvits:

\begin{lem}[\cite{gurvits2005complexity}]\label{lemapp:Gurvits}
Let $A$ be an $m\times m$ matrix. Then, $\Per A$ may be estimated classically with additive precision $\pm\epsilon\|A\|^m$ with high probability, in time $O(\frac{m^2}{\epsilon^2})$, where $\|A\|$ is the largest singular value of $A$.
\end{lem}

\noindent This algorithm has been refined in the case of matrices with repeated lines or repeated columns: 

\begin{lem}[\cite{aaronson2012generalizing}]\label{lemapp:repeated}
Let $B$ be an $m\times n$ matrix. Let $\bm q=(q_1,\dots,q_m)\in\Phi_{m,n}$ and let $A=B_{\bm q,\bm1^n}$ be the $n\times n$ matrix obtained from $B$ by repeating $q_i$ times its $i^{th}$ line. Then, $\Per A$ may be estimated classically with additive precision $\pm\epsilon\cdot\frac{q_1!\dots q_m!}{\sqrt{q_1^{q_1}\cdots q_m^{q_m}}}\|B\|^m$ with high probability, in time $O(\frac{mn}{\epsilon^2})$, where $\|B\|$ is the largest singular value of $B$. Moreover,
\be
|\Per A|\le\frac{q_1!\dots q_m!}{\sqrt{q_1^{q_1}\cdots q_m^{q_m}}}\|B\|^m.
\label{boundperA}
\ee
\end{lem}

\noindent This lemma also allows to approximate efficiently $|\Per A|^2$. Indeed, let $0<\epsilon<1$ and let $z$ be an estimate of $\Per A$ with additive error $\pm\epsilon\cdot\frac{q_1!\dots q_m!}{\sqrt{q_1^{q_1}\cdots q_m^{q_m}}}\|B\|^m$, then $|z|^2$ is a good estimate of $|\Per A|^2$:
\be
\ba
\left||z|^2-|\Per A|^2\right|&=\left||z|-|\Per A|\right|\cdot\left(|z|+|\Per A|\right)\\
&\le\epsilon\cdot\frac{q_1!\dots q_m!}{\sqrt{q_1^{q_1}\cdots q_m^{q_m}}}\|B\|^m\cdot\left(|\Per A|+\epsilon\cdot\frac{q_1!\dots q_m!}{\sqrt{q_1^{q_1}\cdots q_m^{q_m}}}+|\Per A|\right)\\
&\le\epsilon\cdot\frac{q_1!\dots q_m!}{\sqrt{q_1^{q_1}\cdots q_m^{q_m}}}\|B\|^m\cdot\left(\epsilon\cdot\frac{q_1!\dots q_m!}{\sqrt{q_1^{q_1}\cdots q_m^{q_m}}}+2\frac{q_1!\dots q_m!}{\sqrt{q_1^{q_1}\cdots q_m^{q_m}}}\|B\|^m\right)\\
&\le3\epsilon\cdot\frac{q_1!^2\dots q_m!^2}{q_1^{q_1}\cdots q_m^{q_m}}\|B\|^{2m},
\ea
\ee
where we used Eq.~\eqref{boundperA} in the third line. 

In particular, when $B$ is an $m\times n$ submatrix of a unitary matrix (implying $\|B\|\le1$), we may compute an estimate of $|\Per A|^2$, where $A=B_{\bm q,\bm1^n}$ is the $n\times n$ matrix obtained from $B$ by repeating $q_i$ times its $i^{th}$ line, with additive precision $\pm\epsilon\cdot\frac{q_1!^2\dots q_m!^2}{q_1^{q_1}\cdots q_m^{q_m}}$ with high probability, in time $O(\frac{mn}{\epsilon^2})$. The matrices in Eq.~\eqref{probaappfinal} are submatrices of unitary matrices, with repeated lines. Hence, estimating independently all the terms in the sum in Eq.~\eqref{probaappfinal}, we obtain, in time $O(\frac{mn}{\epsilon^2}\cdot|\Phi_{k,r}|)$ and with high probability, an estimate $\tilde P$ of the probability $\text{Pr}_{m,n}^{\text{final}}[\bm s]$ such that
\be
\ba
\left|\tilde P-\text{Pr}_{m,n}^{\text{final}}[\bm s]\right|&\le\epsilon\cdot\frac1{s_1!\dots s_{m-k}!}\sum_{\bm p\in\Phi_{k,r}}{\frac{1}{p_1!\dots p_k!}\cdot\frac{p_1!^2\dots p_k!^2s_1!^2\dots s_{m-k}!^2}{p_1^{p_1}\cdots p_k^{p_k}s_1^{s_1}\cdots s_{m-k}^{s_{m-k}}}}\\
&\le\epsilon\cdot\frac{s_1!\dots s_{m-k}!}{s_1^{s_1}\cdots s_{m-k}^{s_{m-k}}}\sum_{\bm p\in\Phi_{k,r}}{\frac{p_1!\dots p_k!}{p_1^{p_1}\cdots p_k^{p_k}}}\\
&\le\epsilon\cdot|\Phi_{k,r}|,
\ea
\ee
where we used that for all $n\in\mathbb N$, $n!\le n^n$. Note that a tighter bound may be obtained, but the above one is sufficient for our needs. 
In particular, when $|\Phi_{k,r}|=O(\poly m)$, the above procedure provides a polynomially precise additive estimate of the probability $\text{Pr}_{m,n}^{\text{final}}[\bm s]$ in time $O(\poly m)$, with high probability. 

We have
\be
\ba
|\Phi_{k,r}|&=\binom{k+r-1}r\\
&=\frac k{k+r}\cdot\frac{(k+r)!}{k!r!}.
\ea
\ee
This quantity is polynomial in $k$ (resp.\ in $r$) when $r=O(1)$ (resp.\ $k=O(1)$). Moreover, 
\be
\ba
\binom{k+r-1}r&\le\sum_{j=0}^{k+r-1}\binom{k+r-1}j\\
&=2^{k+r-1},
\ea
\ee
so for $k=O(\log m)$ and $r=O(\log m)$, we have $|\Phi_{k,r}|=O(\poly m)$. For all other cases, i.e.\ $k=O(\log m)$ and $r=O(m)$, or $k=O(m)$ and $r=O(\log m)$, or $k=O(m)$ and $r=O(m)$, $|\Phi_{k,r}|$ is superpolynomial in $m$, using Stirling's equivalent $n!\sim\sqrt{2\pi n}(\frac ne)^n$.


\section{Proof of Lemma~\ref{lem:overlapS}}
\label{app:prooflemma2}
\setcounter{lem}{1}

We recall Lemma~\ref{lem:overlapS} from the main text:

\setcounter{lem}{0}

\begin{lem}\label{lemapp:overlapS}
Let $r\in\mathbb N$. The inner product of two (not normalised) output states $\ket{\psi_{\bm p}}$ and $\ket{\psi_{\bm q}}$ of $m$-mode interferometers $U^{\bm p}$ and $V^{\bm q}$ with adaptive measurements outcome $\bm p,\bm q\in\Phi_{k,r}$ is given by
\be
\braket{\psi_{\bm p}|\psi_{\bm q}}=\frac1{\sqrt{\bm p!\bm q!}}\sum_{\substack{\bm i,\bm j\in\{0,1\}^n\\|\bm i|=|\bm j|=r}}{\Per\left(A^{\bm i}\right)\Per\left(B^{\bm j}\right)\Per\left(C^{\bm i,\bm j}\right)},
\ee
where for all $\bm i,\bm j\in\{0,1\}^n$ such that $|\bm i|=|\bm j|=r$,
\be
A^{\bm i}=U^{\bm p\dag}_{(\bm i,\bm0^{m-n}),(\bm p,\bm0^{m-k})}
\ee
is an $r\times r$ matrix which can be obtained efficiently from $U^{\bm p}$,
\be
B^{\bm j}=V^{\bm q}_{(\bm q,\bm0^{m-k}),(\bm j,\bm0^{m-n})}
\ee
is an $r\times r$ matrix which can be obtained efficiently from $V^{\bm q}$, and
\be
C^{\bm i,\bm j}=U^{\bm p\dag}_{(\bm1^n-\bm i,\bm0^{m-n}),(\bm0^k,\bm1^{m-k})}V^{\bm q}_{(\bm0^k,\bm1^{m-k}),(\bm1^n-\bm j,\bm0^{m-n})}
\ee
is an $(n-r)\times(n-r)$ matrix which can be obtained efficiently from $U^{\bm p}$ and $V^{\bm q}$.
\end{lem}

\begin{proof}

Consider the expression for the inner product obtained in Eq.~\eqref{overlapS}:
\be
\braket{\psi_{\bm p}|\psi_{\bm q}}=\frac1{\sqrt{\bm p!\bm q!}}\sum_{\bm s\in\Phi_{m-k,n-r}}{\frac1{\bm s!}\Per\left(U^{\bm p\dag}_{\bm t,(\bm p,\bm s)}\right)\Per\left(V^{\bm q}_{(\bm q,\bm s),\bm t}\right)}.
\label{overlapSapp}
\ee
It is reminiscent of the permanent composition formula~\cite{percus2012combinatorial}: for all $m,n,c\in\mathbb N^*$, all $s\in\mathbb N$, all $\bm u\in\Phi_{m,s}$ and all $\bm v\in\Phi_{n,s}$,
\be
\Per\left[(MN)_{\bm u,\bm v}\right]=\sum_{\bm s\in\Phi_{c,s}}{\frac1{\bm s!}\Per\left(M_{\bm u,\bm s}\right)\Per\left(N_{\bm s,\bm v}\right)}
\label{percomp2}
\ee
where $M$ is a $m\times c$ matrix and $N$ is a $n\times c$ matrix. However, this formula is not directly applicable to the expression in Eq.~\eqref{overlapSapp}. In order to obtain a suitable expression, we first make use of the Laplace formula for the permanent: we expand the permanent of $U^{\bm p\dag}_{\bm t,(\bm p,\bm s)}$ along the columns that are repeated according to $\bm p$ and we expand the permanent of $V^{\bm q}_{(\bm q,\bm s),\bm t}$ along the rows that are repeated according to $\bm q$.
The generalised Laplace column expansion formula for the permanent reads: let $n\in\mathbb N^*$, let $W$ be an $n\times n$ matrix, and let $\bm j\in\{0,1\}^n$. Then,
\be
\Per\,(W)=\sum_{\substack{\bm i\in\{0,1\}^n\\|\bm i|=|\bm j|}}{\Per\left(W_{\bm i,\bm j}\right)\Per\left(W_{\bm1^n-\bm i,\bm1^n-\bm j}\right)},
\label{Laplacegen}
\ee
where $W_{\bm i,\bm j}$ is the matrix obtained from $W$ by keeping only the $k^{th}$ rows and $l^{th}$ columns such that $i_k=1$ and $j_l=1$, respectively, and $W_{\bm1^n-\bm i,\bm1^n-\bm j}$ is the matrix obtained from $W$ by keeping only the $k^{th}$ rows and $l^{th}$ columns such that $i_k=0$ and $j_l=0$, respectively. This formula is obtained by applying the Laplace expansion formula for one column various times, for each column with index $l$ such that $j_l=1$, and the same formula holds for rows.

We first apply the general column expansion formula in Eq.~\eqref{Laplacegen} to the matrix $U^{\bm p\dag}_{\bm t,(\bm p,\bm s)}$ with $\bm j=(\bm1^r,\bm0^{n-r})\in\{0,1\}^n$, obtaining
\be
\Per\left(U^{\bm p\dag}_{\bm t,(\bm p,\bm s)}\right)=\sum_{\substack{\bm i\in\{0,1\}^n\\|\bm i|=r}}{\Per\left[\left(U^{\bm p\dag}_{\bm t,(\bm p,\bm s)}\right)_{\bm i,\bm j}\right]\Per\left[\left(U^{\bm p\dag}_{\bm t,(\bm p,\bm s)}\right)_{\bm1^n-\bm i,\bm1^n-\bm j}\right]}.
\label{expandU}
\ee
Let us consider the matrix $\left(U^{\bm p\dag}_{\bm t,(\bm p,\bm s)}\right)_{\bm i,\bm j}$ appearing in this last expression, for $\bm i\in\{0,1\}^n$. Its rows are obtained by keeping the first $n$ lines of $U^{\bm p\dag}$ since $\bm t=(\bm1^n,\bm0^{m-n})$, then by keeping only the $l^{th}$ rows such that $i_l=1$. Its columns are obtained by repeating $p_l$ times the $l^{th}$ column for $l\in\{1,\dots,k\}$ and $s_l$ times for $l\in\{k+1,\dots,m\}$, then by only keeping the first $r$ columns since $\bm j=(\bm1^r,\bm0^{n-r})$. However, since $|\bm p|=|\bm j|=r$, these are the columes repeated according to $\bm p$. Hence,
\be
\left(U^{\bm p\dag}_{\bm t,(\bm p,\bm s)}\right)_{\bm i,\bm j}=U^{\bm p\dag}_{(\bm i,\bm0^{m-n}),(\bm p,\bm0^{m-k})},
\label{simplerU1}
\ee
where $U^{\bm p\dag}_{(\bm i,\bm0^{m-n}),(\bm p,\bm0^{m-k})}$ is the matrix obtained from $U^{\bm p\dag}$ by keeping only the $l^{th}$ rows such that $i_l=1$ and removing the others, and by repeating $p_l$ times the $l^{th}$ column for $l\in\{1,\dots,k\}$ and removing the others. Similarly, with $|\bm s|=|\bm1^n-\bm j|=n-r$,
\be
\left(U^{\bm p\dag}_{\bm t,(\bm p,\bm s)}\right)_{\bm1^n-\bm i,\bm1^n-\bm j}=U^{\bm p\dag}_{(\bm1^n-\bm i,\bm0^{m-n}),(\bm0^k,\bm s)},
\label{simplerU2}
\ee
where $U^{\bm p\dag}_{(\bm1^n-\bm i,\bm0^{m-n}),(\bm0^k,\bm s)}$  is the matrix obtained from $U^{\bm p\dag}$ by keeping only the $l^{th}$ rows such that $i_l=0$ and removing the others, and by repeating $s_l$ times the $l^{th}$ column for $l\in\{k+1,\dots,m\}$ and removing the others. With Eqs.~\eqref{expandU}, \eqref{simplerU1} and \eqref{simplerU2} we obtain
\be
\ba
\Per\left(U^{\bm p\dag}_{\bm t,(\bm p,\bm s)}\right)&=\sum_{\substack{\bm i\in\{0,1\}^n\\|\bm i|=r}}{\Per\left(U^{\bm p\dag}_{(\bm i,\bm0^{m-n}),(\bm p,\bm0^{m-k})}\right)\Per\left(U^{\bm p\dag}_{(\bm1^n-\bm i,\bm0^{m-n}),(\bm0^k,\bm s)}\right)}\\
&=\sum_{\substack{\bm i\in\{0,1\}^n\\|\bm i|=r}}{\Per\left(A^{\bm i}\right)\Per\left(U^{\bm p\dag}_{(\bm1^n-\bm i,\bm0^{m-n}),(\bm0^k,\bm s)}\right)},
\ea
\label{expandU2}
\ee
where we have defined, for all $\bm i\in\{0,1\}^n$ such that $|\bm i|=r$,
\be
A^{\bm i}:=U^{\bm p\dag}_{(\bm i,\bm0^{m-n}),(\bm p,\bm0^{m-k})},
\ee
which is an $r\times r$ matrix independent of $\bm s$ that can be obtained efficiently from $U^{\bm p}$.

The same reasoning with the general row expansion formula for the matrix $V^{\bm q}_{(\bm q,\bm s),\bm t}$ and the rows $\bm i=(\bm1^r,\bm0^{n-r})$ gives
\be
\ba
\Per\left(V^{\bm q}_{(\bm q,\bm s),\bm t}\right)&=\sum_{\substack{\bm j\in\{0,1\}^n\\|\bm j|=r}}{\Per\left[\left(V^{\bm q}_{(\bm q,\bm s),\bm t}\right)_{\bm i,\bm j}\right]\Per\left[\left(V^{\bm q}_{(\bm q,\bm s),\bm t}\right)_{\bm1^n-\bm i,\bm1^n-\bm j}\right]}\\
&=\sum_{\substack{\bm j\in\{0,1\}^n\\|\bm j|=r}}{\Per\left(V^{\bm q}_{(\bm q,\bm0^{m-k}),(\bm j,\bm0^{m-n})}\right)\Per\left(V^{\bm q}_{(\bm0^k,\bm s),(\bm1^n-\bm j,\bm0^{m-n})}\right)},
\ea
\label{expandV}
\ee
where $V^{\bm q}_{(\bm q,\bm0^{m-k}),(\bm j,\bm0^{m-n})}$ is the matrix obtained from $V^{\bm q}$ by repeating $q_l$ times the $l^{th}$ row for $l\in\{1,\dots,k\}$ and removing the others and by keeping only the $l^{th}$ columns such that $j_l=1$, and where $V^{\bm q}_{(\bm0^k,\bm s),(\bm1^n-\bm j,\bm0^{m-n})}$ is the matrix obtained from $V^{\bm q}$ by repeating $s_l$ times the $l^{th}$ row for $l\in\{k+1,\dots,m\}$ and removing the others and by keeping only the $l^{th}$ columns such that $j_l=0$.
Defining, for all $\bm j\in\{0,1\}^n$ such that $|\bm j|=r$,
\be
B^{\bm j}:=V^{\bm q}_{(\bm q,\bm0^{m-k}),(\bm j,\bm0^{m-n})},
\ee
the expression in Eq.~\eqref{expandV} rewrites
\be
\Per\left(V^{\bm q}_{(\bm q,\bm s),\bm t}\right)=\sum_{\substack{\bm j\in\{0,1\}^n\\|\bm j|=r}}{\Per\left(B^{\bm j}\right)\Per\left(V^{\bm q}_{(\bm0^k,\bm s),(\bm1^n-\bm j,\bm0^{m-n})}\right)},
\label{expandV2}
\ee
where $B^{\bm j}$ are $r\times r$ matrices independent of $\bm s$ and can be obtained efficiently from $V^{\bm q}$.

Plugging Eqs.~\eqref{expandU2} and \eqref{expandV2} in Eq.~\eqref{overlapSapp} we obtain
\be
\ba
\braket{\psi_{\bm p}|\psi_{\bm q}}&=\frac1{\sqrt{\bm p!\bm q!}}\sum_{\substack{\bm i,\bm j\in\{0,1\}^n\\|\bm i|=|\bm j|=r}}\Bigg[\Per\left(A^{\bm i}\right)\Per\left(B^{\bm j}\right)\\
&\quad\quad\quad\times\sum_{\bm s\in\Phi_{m-k,n-r}}{\frac1{\bm s!}\Per\left(U^{\bm p\dag}_{(\bm1^n-\bm i,\bm0^{m-n}),(\bm0^k,\bm s)}\right)\Per\left(V^{\bm q}_{(\bm0^k,\bm s),(\bm1^n-\bm j,\bm0^{m-n})}\right)}\Bigg].
\ea
\label{overlapS2}
\ee
The sum appearing in the second line may now be expressed as a single permanent using the permanent composition formula: for all $\bm i,\bm j\in\{0,1\}^n$ such that $|\bm i|=|\bm j|=r$, let us define the $(n-r)\times(m-k)$ matrix
\be
\tilde U^{\bm p,\bm i}:=U^{\bm p\dag}_{(\bm1^n-\bm i,\bm0^{m-n}),(\bm0^k,\bm1^{m-k})},
\ee
and the $(m-k)\times(n-r)$ matrix
\be
\tilde V^{\bm q,\bm j}:=V^{\bm q}_{(\bm0^k,\bm1^{m-k}),(\bm1^n-\bm j,\bm0^{m-n})},
\ee
so that
\be
U^{\bm p\dag}_{(\bm1^n-\bm i,\bm0^{m-n}),(\bm0^k,\bm s)}=\tilde U^{\bm p,\bm i}_{\bm1^{n-r},\bm s}\quad\text{and}\quad V^{\bm q}_{(\bm0^k,\bm s),(\bm1^n-\bm j,\bm0^{m-n})}=\tilde V^{\bm q,\bm j}_{\bm s,\bm1^{n-r}}.
\ee
With the permanent composition formula in Eq.~\eqref{percomp2} we obtain
\be
\sum_{\bm s\in\Phi_{m-k,n-r}}{\frac1{\bm s!}\Per\left(U^{\bm p\dag}_{(\bm1^n-\bm i,\bm0^{m-n}),(\bm0^k,\bm s)}\right)\Per\left(V^{\bm q}_{(\bm0^k,\bm s),(\bm1^n-\bm j,\bm0^{m-n})}\right)}=\Per\left[\left(\tilde U^{\bm p,\bm i}\tilde V^{\bm q,\bm j}\right)_{\bm1^{n-r},\bm1^{n-r}}\right].
\ee
Since $\tilde U^{\bm p,\bm i}\tilde V^{\bm q,\bm j}$ is an $(n-r)\times(n-r)$ matrix we thus have
\be
\sum_{\bm s\in\Phi_{m-k,n-r}}{\frac1{\bm s!}\Per\left(U^{\bm p\dag}_{(\bm1^n-\bm i,\bm0^{m-n}),(\bm0^k,\bm s)}\right)\Per\left(V^{\bm q}_{(\bm0^k,\bm s),(\bm1^n-\bm j,\bm0^{m-n})}\right)}=\Per\left(\tilde U^{\bm p,\bm i}\tilde V^{\bm q,\bm j}\right).
\ee
Then, Eq.~\eqref{overlapS2} rewrites
\be
\braket{\psi_{\bm p}|\psi_{\bm q}}=\frac1{\sqrt{\bm p!\bm q!}}\sum_{\substack{\bm i,\bm j\in\{0,1\}^n\\|\bm i|=|\bm j|=r}}{\Per\left(A^{\bm i}\right)\Per\left(B^{\bm j}\right)\Per\left(C^{\bm i,\bm j}\right)},
\label{overlapS3}
\ee
where we have defined
\be
\ba
C^{\bm i,\bm j}&:=\tilde U^{\bm p,\bm i}\tilde V^{\bm q,\bm j}\\
&=U^{\bm p\dag}_{(\bm1^n-\bm i,\bm0^{m-n}),(\bm0^k,\bm1^{m-k})}V^{\bm q}_{(\bm0^k,\bm1^{m-k}),(\bm1^n-\bm j,\bm0^{m-n})},
\ea
\ee
which is an $(n-r)\times(n-r)$ matrix that can be computed efficiently from $U^{\bm p}$ and $V^{\bm q}$.

\end{proof}


\end{document}